\documentclass[a4paper,uplatex]{article}
\usepackage{fullpage}
\usepackage{amssymb,amsmath}
\usepackage{graphicx}
\usepackage{hyperref}

\usepackage{amsfonts}
\usepackage{color}
\usepackage{url}
\usepackage{colortbl,array,xcolor}

\newcommand{\qed}{\rule[1pt]{4pt}{6pt}}

\newtheorem{theorem}{Theorem}
\newtheorem{lemma}[theorem]{Lemma}

\newtheorem{corollary}[theorem]{Corollary}

\newenvironment{proof}{\par\noindent{{\sf Proof}}}{\hfill\qed\par}

\newcommand{\drawS}{\begin{tabular}[c]{c}\includegraphics[width=6mm,bb=0 0 55 54]{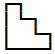}\end{tabular}}
\newcommand{\drawJ}{\begin{tabular}[c]{c}\includegraphics[width=6mm,bb=0 0 55 54]{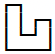}\end{tabular}}
\newcommand{\drawF}{\begin{tabular}[c]{c}\includegraphics[width=8mm,bb=0 0 66 38]{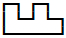}\end{tabular}}

\newcommand{\drawJRec}{\begin{tabular}[c]{c}\includegraphics[width=8mm,bb=0 0 66 54]{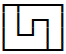}\end{tabular}}

\newcommand{\drawJX}{\begin{tabular}[c]{c}\includegraphics[width=8mm,bb=0 0 66 68]{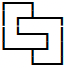}\end{tabular}}

\newcommand{\drawFt}{\begin{tabular}[c]{c}\includegraphics[width=14mm,bb=0 0 106 82]{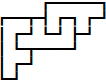}\end{tabular}}

\newcommand{\drawFs}{\begin{tabular}[c]{c}\includegraphics[width=12mm,bb=0 0 95 82]{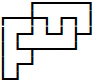}\end{tabular}}

\newcommand{\drawSz}{\begin{tabular}[c]{c}\includegraphics[width=8mm,bb=0 0 68 55]{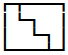}\end{tabular}}

\newcommand{\drawSo}{\begin{tabular}[c]{c}\includegraphics[width=10mm,bb=0 0 80 68]{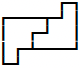}\end{tabular}}

\newcommand{\drawSt}{\begin{tabular}[c]{c}\includegraphics[width=12mm,bb=0 0 96 98]{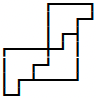}\end{tabular}}

\newcommand\YY{$(\checkmark)$}
\newcommand\NN{$(\times)$}

\begin{document}

\title{Solving Rep-tile by Computers: Performance of Solvers and Analyses of Solutions}

\author{
  \begin{tabular}{c@{\qquad}c@{\qquad}c@{\qquad}c}
  Mutsunori Banbara\thanks{Nagoya University, Japan. {\tt \{banbara,k-hasimt,sakai\}@i.nagoya-u.ac.jp}}
  &
  Kenji Hashimoto\footnotemark[1]
  &
  Takashi Horiyama\thanks{Hokkaido University, Japan. {\tt horiyama@ist.hokudai.ac.jp}}
  &
  Shin-ichi Minato\thanks{Kyoto University, Japan. {\tt minato@i.kyoto-u.ac.jp}}
  \\[\medskipamount]
  Kakeru Nakamura\thanks{Japan Advanced Institute of Science and Technology, Japan. {\tt \{s2010142,uehara\}@jaist.ac.jp}}
  &
  Masaaki Nishino\thanks{NTT Corporation, Japan. {\tt \{masaaki.nishino.uh,norihito.yasuda.hn\}@hco.ntt.co.jp}}
  &
  Masahiko Sakai\footnotemark[1]
  &
  Ryuhei Uehara\footnotemark[4]
  \\[\medskipamount]
  &
  Yushi Uno\thanks{Osaka Prefecture University, Japan. {\tt uno@cs.osakafu-u.ac.jp}}
  &
  Norihito Yasuda\footnotemark[5]
  \end{tabular}
}

\index{Banbara, Mutsunori}
\index{Hashimoto, Kenji}
\index{Horiyama, Takashi}
\index{Minato, Shin-ichi}
\index{Nakamura, Kakeru}
\index{Nishino, Masaaki}
\index{Sakai, Masahiko}
\index{Uehara, Ryuhei}
\index{Uno, Yushi}
\index{Yasuda, Norihito}

\maketitle

\begin{abstract}
A \emph{rep-tile} is a polygon that can be dissected into smaller copies (of the same size) of the original polygon.
A \emph{polyomino} is a polygon that is formed by joining one or more unit squares edge to edge.
These two notions were first introduced and investigated by Solomon W.~Golomb in the 1950s and popularized by Martin Gardner in the 1960s.
Since then, dozens of studies have been made in communities of recreational mathematics and puzzles.
In this study, we first focus on the specific rep-tiles that have been investigated in these communities.
Since the notion of rep-tiles is so simple that can be formulated mathematically in a natural way,
we can apply a representative puzzle solver, a MIP solver, and SAT-based solvers
for solving the rep-tile problem in common.
In comparing their performance, we can conclude that the puzzle solver is the weakest
while the SAT-based solvers are the strongest in the context of simple puzzle solving.
We then turn to analyses of the specific rep-tiles.
Using some properties of the rep-tile patterns found by a solver,
 we can complete analyses of specific rep-tiles up to certain sizes.
That is, up to certain sizes, we can determine the existence of solutions,
clarify the number of the solutions, or we can enumerate all the solutions for each size.
In the last case, we find new series of solutions for the rep-tiles which have never been found in the communities.

\end{abstract}

\thispagestyle{empty}


\section{Introduction}
\label{sec:intro}

In some games like Tetris, polygons obtained by joining unit squares edge to edge are used as their pieces.
These polygons are called polyominoes, and they have been used in popular puzzles since at least 1907.
Solomon W. Golomb introduced the name polyomino in 1953 and was widely investigated \cite{Golomb}.
It was popularized in the 1960s by the famous column in \emph{Scientific American} written by Martin Gardner \cite{Gardner1}.

\begin{figure}[ht]
\centering
\includegraphics[width=0.3\textwidth]{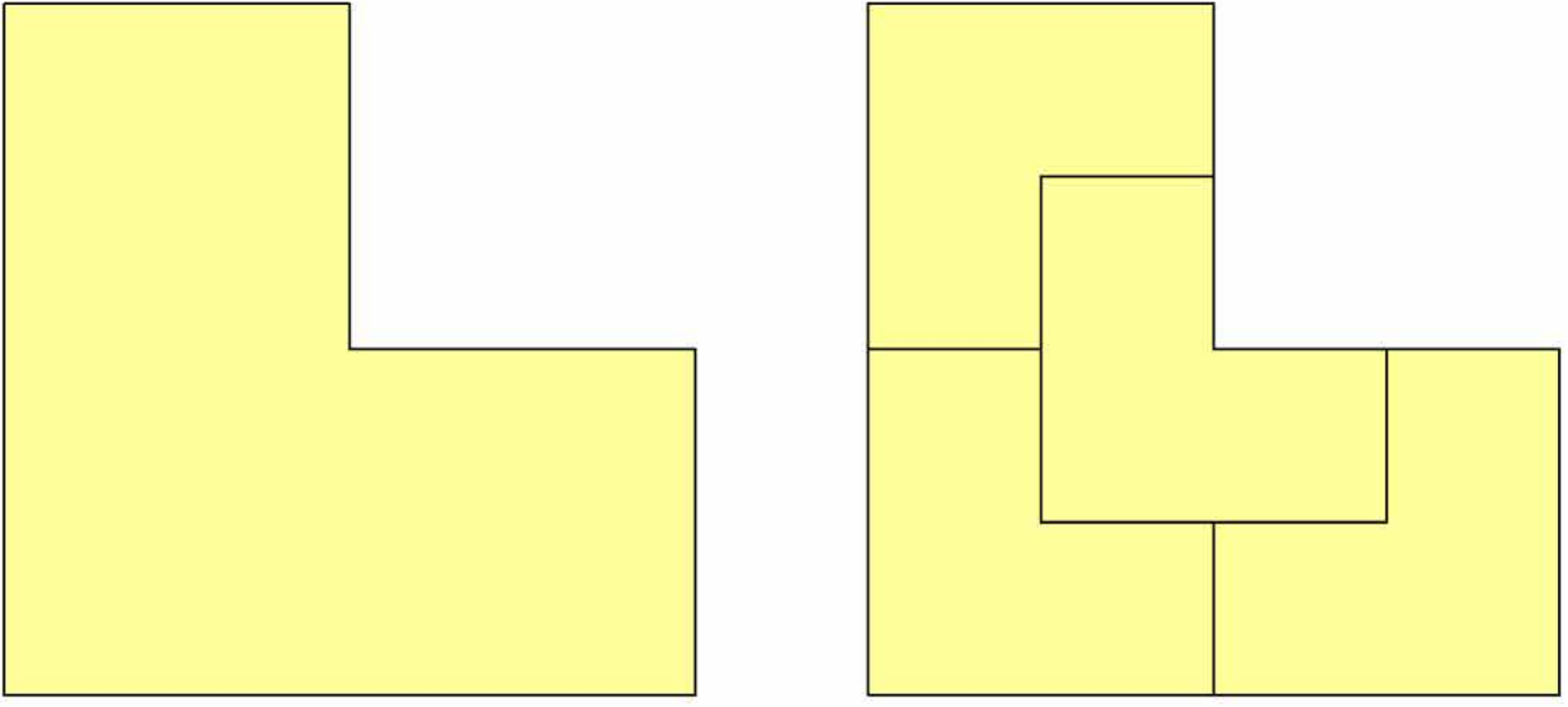}
\caption{An example of a rep-tile of rep-4.}
\label{fig:rep1}
\end{figure}

Golomb is also known as an inventor of the notion of rep-tile.
A polygon $P$ is called rep-tile if it can be dissected into smaller copies of $P$.
Especially, if $P$ can be dissected into $n$ copies, it is said to be rep-$n$.
An example of a rep-tile of rep-4 is given in \figurename~\ref{fig:rep1}.
We can observe that each of 4 copies can be dissected into 4 smaller copies, which give us rep-$16$.
That is, a rep-tile of rep-$n$ is also rep-$n^i$ for any positive integer $i=1,2,\ldots$.
We also extend the rep-tile of rep-$n$ by tiling $n$ copies to make a larger pattern.
That is, we can tile the plane by repeating this process.
It is known that some rep-tile can be used to generate acyclic tiling
(i.e., the tiling pattern cannot be identical by shifting and rotation).
Both cyclic and acyclic tilings have been well investigated since 
they have applications to chemistry, especially, crystallography \cite{Gardner4}.
From the viewpoints of mathematics and art, the notion of rep-tile is popular 
as we can obtain tiling of the plane with the same shapes of different sizes 
by replacing a part of the rep-tiles by their copies recursively.

\begin{figure}[ht]
\centering
\includegraphics[width=0.4\textwidth]{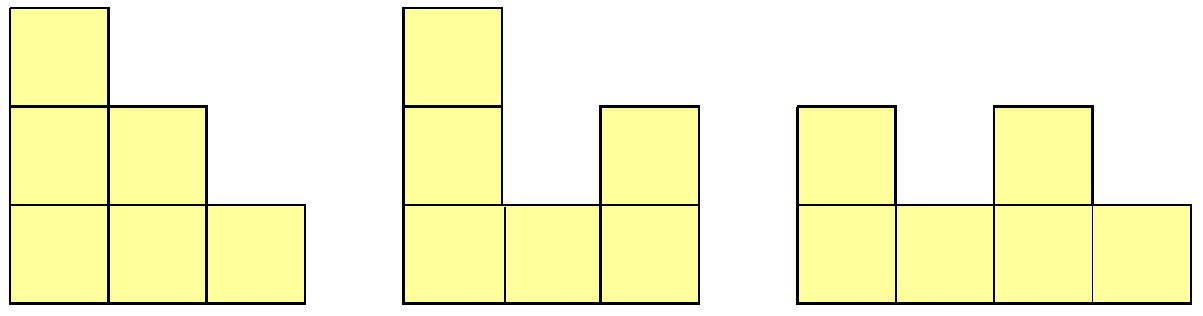}
\caption{The 6-ominoes of \emph{stair-shape}, \emph{J-shape}, and \emph{F-shape}}
\label{fig:rep144}
\end{figure}

Gardner introduced the polyomino rep-tiles in \cite{Gardner4}.
Precisely, he introduced three 6-ominoes (polyominoes formed by 6 unit squares) 
in \figurename~\ref{fig:rep144} as rep-tiles of rep-144.
When the article \cite{Gardner4} was written, 
the minimum number of dissections for these three rep-tiles was conjectured as 144.
Namely, they are the rep-tiles of rep-144, and not rep-$k$ for any $1<k<144$.
However, they have been found out that the left \emph{stair-shape} is a rep-tile of rep-121,
the central \emph{J-shape} is a rep-tile of rep-36, and
the right \emph{F-shape} is a rep-tile of rep-64 \cite{Gardner4,RepPage}.

Polyomino rep-tiles have a long history mainly in the contexts of puzzles and recreational mathematics.
They have been investigated since the 1950s, however, they have relied on discoveries by hand.
In fact, there are many constructive solutions for these puzzles on the web page \cite{RepPage}
However, these puzzles have not yet ``solved'' in the strict sense 
since any nonexistent results for these cases have not yet be given.

In this research, we first experiment on these three polyominoes in \figurename~\ref{fig:rep144}
and check if they are rep-$n$ for each $n$ by the representative approaches by a computer.
Since the notion of a rep-tile is a quite simple puzzle,
we can represent the conditions of a rep-tile in several different natural ways
in the terms of representative problem solvers.
Therefore, we can compare the performance of the different problem solvers using such a simple puzzle as a common problem.
We use the following three different approaches to solving the rep-tiles by a computer.
\begin{description}
\item[Puzzle solver and implementation based on dancing links:]
Nowadays, most puzzle designers use a free puzzle solver.
It is based on a data structure called dancing links proposed by Knuth.
It is said that dancing links is the data structure that allows us to perform backtracking efficiently,
and hence it is suitable to analyze puzzles.
Although we do not know the details of the implementation of the free puzzle solver,
we also independently implemented two algorithms; one uses dancing links,
and the other uses dancing links with ZDD to make it faster.

\item[MIP solver:] When we formalize the solutions of a rep-tile
by constraint integer programming, we can solve it by mixed integer programming (MIP) solvers.
The conditions of a rep-tile can be formalized in a relatively simple integer programming (IP),
and we can decide if the rep-tile has a solution if and only if the corresponding instance 
in the form of the IP is feasible. Since each feasible solution corresponds to a solution, 
the number of feasible solutions also gives the number of solutions of the rep-tile.
In this formulation, the feasibility is the issue and hence the optimization term in the MIP solver is redundant.

\item[SAT-based solver:] 
Most instances of the integer programming can be solved by SAT-based solvers with some modifications of constraints.
It is the case for the conditions of a rep-tile,
and hence the IP formulation can be translated to the constraints of the SAT-based solvers.
\end{description}
In summary, the puzzle solver and programs based on dancing links, even if we use ZDD,
cannot solve rep-tiles of rep-$n$ for large $n$.
However, this fact does not mean the limit of using a computer.
The MIP solver can solve rep-tiles of rep-$n$ for larger $n$ than the puzzle solvers.
Moreover, we found out that the SAT-based solvers can solve much larger sizes than the MIP solver.
These results were contrary to our expectations.

By using a model counting method with a SAT-based solver,
we succeeded to count the number of solutions of rep-tiles of certain sizes, 
which are bigger than the previously known results.
Our results are summarized in Table \ref{tab:rep-tile}.
(As we will describe later,
there exist $n$-omino rep-tiles only when $n=k^2$ for some positive integer $k$.
Therefore, we will consider $k^2$-omino rep-tiles for $k=1,2,\ldots$.)

\begin{table*}
\begin{minipage}[t]{0.8\textwidth}
\begin{tabular}[t]{|c|r|r|r|r|r|r|r|r|r|r|r|}\hline
\rowcolor{lightgray} $k$    & 1 & 2 & 3 & 4 & 5 & 6 & 7 & 8 & 9 & 10 & 11\\\hline
\drawS & 1 & 0 & 0 & 0 & 0 & 0 & 0 & 0 & 0 &  0 & 32858262881295138816 \\
\drawJ & 1 & 0 & 0 & 0 & 0 & 262144 & 0 & 0 & 0 & 0 & 0\\
\drawF & 1 & 0 & 0 & 0 & 0 & 0 & 0 & 1358954496 & 51539607552 & 0 & 0\\\hline
\end{tabular}

\begin{tabular}[t]{|c|r|r|r|r|r|}\hline
\rowcolor{lightgray} $k$    & 12 & 13 & 14 & 15 & 16 \\\hline
\drawS & 7513742553498633531870412820 & 421105971327597731222250323968 & 0 & 0 & 0 \\
\drawJ & 545409716939029673955819520 & 0 & 0 & 0 & 0 \\
\drawF & 693242756013012824879005696 & 3658830332096120778961977344 & 0 & $>0$ & $>0$ \\\hline
\end{tabular}

\begin{tabular}[t]{|c|r|r|r|r|r|r|r|r|r|r|}\hline
\rowcolor{lightgray} $k$    & 17  &   18  & 19   & 20   & 21  & 22  & 23   & 24   & 25   \\\hline
\drawS &                       0  &   0   &  ?   &  ?   & ?   & ?   & $>0$ & $>0$ & $>0$ \\
\drawJ &                       0  & $>0$  & 0    &  0   & 0   & 0   & 0    & $>0$ & 0 \\
\drawF &                     $>0$ &   ?   & $>0$ & $>0$ & $>0$& ?   & $>0$ & $>0$ & $>0$  \\\hline
\end{tabular}

\end{minipage}
\caption{The number of distinct dissections of $k^2$-omino rep-tiles,
where each number indicates the number of solutions, 
where $0$ means no solution, $>0$ means at least one solution, and $?$ means unknown.}
\label{tab:rep-tile}
\end{table*}

By examining in detail the number of solutions and the specific individual solutions,
we obtain two major new results regarding rep-tiles.

Each 0 in Table \ref{tab:rep-tile} indicates 
that there is no rep-tile of rep-$k^2$ for the corresponding 6-omino.
Since the previously known results of rep-tiles only indicate the existence of a solution constructively, 
it remains open whether there is a solution for other sizes.
In this paper, we show for the first time that there is no solution up to a certain size.
It was conjectured that these three rep-tiles of rep-$144$ were the minimum size in \cite{Gardner4},
and then gradually, smaller solutions were shown constructively.
However, it has never been proved that they are the minimum number.
Our results in Table \ref{tab:rep-tile} reveal for the first time that they are all the minimum rep-tiles.
They put an end to the history of exploration of these rep-tiles for more than 50 years.

As for the size in which solutions exist, we succeed in completely characterizing some of the solutions 
by analyzing the number of solutions and patterns of these solutions.
They contain whole new types of solutions that are not included in previously known constructive solutions.
We also succeed in constructing solutions with completely different characteristics from 
the known solutions by combining a constructive method and a search using these new types of solutions as clues.
By developing these new types of solutions,
it may be possible to find completely new solutions even for 
sizes that are previously expected to have no solution.

\section{Preliminaries}
\label{sec:pre}

A \emph{polyomino} is a simple polygon that can be obtained by joining unit squares edge by edge.
All polygons in this paper are polyominoes.
For an integer $s$, a polyomino of area $s$ is called an \emph{$s$-omino}.
A simple polygon $P$ is a \emph{rep-tile} of \emph{rep-$n$} 
if $P$ can be dissected into $n$ congruent polygons similar to $P$.

In this paper, we focus on polyomino rep-tiles.
Then the following theorem is important.
\begin{theorem}
\label{th:power}
When a polyomino $P$ is a rep-tile of rep-$n$, $n$ is a square number.
That is, there exists a natural number $k$ such that $n=k^2$.
\end{theorem}
\begin{proof}
Let $P$ be a $t$-omino. That is, $P$ consists of $t$ unit squares.
By assumption, $P$ can be dissected to $n$ copies of $P'$, where $P'$ is similar to $P$.
Then, since a unit square has an edge of length 1, the corresponding square of $P'$ has an edge of length $1/\sqrt{n}$.
Let $\ell$ be the length of a shortest edge $e$ of the polyomino $P$.
Then, $\ell$ is an integer and $\ell$ should be a multiple of $1/\sqrt{n}$ since this edge $e$ is formed by tiling $P'$.
Therefore, $\sqrt{n}$ should be an integer, and hence $n$ is a square number.
\end{proof}

By Theorem \ref{th:power}, a polyomino $P$ cannot be a rep-tile of rep-$n$ when $n$ is not a square number.
Therefore, we assume that $n=k^2$ for some positive integer $k$ without loss of generality.
In order to compare to the previous results, we focus on the three 6-ominoes shown in \figurename~\ref{fig:rep144} in this paper.
We call each of them \emph{stair-shape}, \emph{J-shape}, and \emph{F-shape}, respectively.

Among these three 6-ominoes, the J-shape and the F-shape are concave, 
and hence the concave part should be filled by the other piece to construct a rep-tile.
Precisely, a polyomino $P$ is \emph{concave} if there exists a unit square not belonging to $P$ but 
it shares three edges with $P$. We call this square \emph{concave square} of $P$.

In this research, we solve the polyomino rep-tile problem for the three 6-ominoes by some problem solvers.
When we use MIP solver or SAT-based solvers, we have to describe the constraints of the rep-tile problem.
Here we give the common way for the representation.

As a simple example, we consider a domino (or 2-omino) $P$ of rep-4.
In this case, since $4=2^2$ is the square number of $k=2$, 
we consider $P$ as an 8-omino of size $4\times 2$ by scaling 2 and fill $P$ by 4 dominoes of size $2\times 1$.
We first assign a unique number to each unit square of $P$. 
We let
\begin{tabular}[h]{|c|c|c|c|}\hline
0& 1& 2& 3\\\hline
4& 5& 6& 7\\\hline
\end{tabular}
for example.
When we tile 4 dominoes on the 8-omino $P$, 
a binary variable $A(i,j)$ using the numbers of unit squares indicates a way of each domino.
To make the representation unique, we assume that $i<j$.
For example, when $A(0,1)=1$, it means that a domino covers the unit squares 0 and 1.
For this $P$, we use 10 binary variables 
($A(0,1)$, $A(1,2)$, $A(2,3)$, $A(4,5)$, $A(5,6)$, $A(6,7)$, $A(0,4)$, $A(1,5)$, $A(2,6)$, $A(3,7)$)
to represent if a domino covers the corresponding unit squares.

Next, we introduce constraints for each unit square.
Precisely, since each unit square $i$ should be covered by just one domino, 
we have the following constraints.
\begin{eqnarray*}
\mbox{Constraint for the square 0}:& A(0,1)+A(0,4)=1\\
\mbox{Constraint for the square 1}:& A(0,1)+A(1,2)+A(1,5)=1\\
\mbox{Constraint for the square 2}:& A(1,2)+A(2,3)+A(2,6)=1\\
\mbox{Constraint for the square 3}:& A(2,3)+A(3,7)=1\\
\mbox{Constraint for the square 4}:& A(4,5)+A(0,4)=1\\
\mbox{Constraint for the square 5}:& A(4,5)+A(5,6)+A(1,5)=1\\
\mbox{Constraint for the square 6}:& A(5,6)+A(6,7)+A(2,6)=1\\
\mbox{Constraint for the square 7}:& A(6,7)+A(3,7)=1
\end{eqnarray*}
It is clear that $P$ is a rep-tile of rep-$4$ if and only if 
there is a solution that satisfies these eight constraints.

In this paper, we wrote programs that generate 
the declarations of the binary variables and 
the corresponding constraints 
for each combination of 6-ominoes stair-shape, J-shape, or F-shape, and 
a square number $n=k^2$.

\section{Comparisons of Solvers}
\label{sec:exp}

As representative problem solvers, we chose BurrTools as a puzzle solver,
SCIP as a MIP solver, and clingo and NaPS as SAT-based solvers.
For each of the three rep-tiles, we list their running time for solving the rep-tile.
The details and resources of the solvers follow them.
Tables \ref{tab:S}, \ref{tab:J}, \ref{tab:F} summarize the running times of solvers for each rep-tile.
In the tables, DLX indicates the algorithm based on dancing links,
DLZ indicates the algorithm based on dancing links with ZDD, which are implemented by ourselves to compare with BurrTools.
We omit the cases $k<6$ since they are too short, and each number represents seconds.
The symbol ? means timeout in this case.
We set the time limit for each solver as 10 minutes (600 seconds) in DLX/DLZ, 
12 hours (43200 seconds) in clingo, and 2 days (172800 seconds) in NaPS.
The entry OF in DLZ means ``overflow of cache''.
After each $k$, we put $\checkmark$, $\times$, and ? which mean
``there exists a solution'', ``there exists no solution'', and
``we do not know if there is a solution or not,'' respectively.

\begin{table*}[ht]
\centering
\begin{tabular}[t]{|c|r|r|r|r|r|r|r|r|r|r|r|r|}\hline
\rowcolor{lightgray} $k$ (Solution?) & 6\NN & 7\NN    & 8\NN    & 9\NN    &  10\NN & 11\YY  & 12\YY & 13\YY & 14\NN  &  15\NN \\\hline
BurrTools 0.6.3 & $<1$ & 2 & 5760 & ? & ? & ? & ? & ? & ? & ? \\\hline
DLX(1st solution) & $<1$   & 15 & ?  & ?  & ? & ? & $<1$ & ? & ? & ? \\
DLX(all solutions) & ?    & ?  & ?  & ?  & ? & ? & ?  & ? & ? & ? \\\hline
DLZ(1st solution) & $<1$  & $<1$ & $<1$ & 391  & OF   & OF  & $<1$ & OF  & OF  & OF \\
DLZ(all solutions) & ?   & ?  & ?  & ?  & ?  & ? & OF   & OF  & ? & ? \\\hline
SCIP 7.0.2    & 1 & 1 & 1 & 1 & 56 & 42 & 7 & 120 & ? & ? \\\hline
clingo 5.4.0  & $<1$ & $<1$ & $<1$ & $<1$ & $<1$ & 1 & 2 & 2 & 8 & ? \\\hline
NaPS 1.02b2   & $<1$ & $<1$ & $<1$ & $<1$ & $<1$ & $<1$ & $<1$ & 1  & 17  & 6388 \\\hline\hline
\rowcolor{lightgray} $k$ (Solution?) & 16\NN & 17\NN & 18\NN & 19(?)  & 20(?)    & 21(?)  & 22(?)  & 23\YY  & 24\YY  & 25\YY\\\hline
clingo 5.4.0  & ? & 2946 & ? & ? & ? & ?  & ? & 8911 & 26973 & ? \\\hline
NaPS 1.02b2   & (421700) & 1163 & 12530 & ? & ?   & ? & ? & 529 & 1415 & 1744  \\\hline
\end{tabular}
\caption{Time (sec.) for deciding if \drawS{} is a rep-tile of rep-$k^2$
(NaPS finishes its computation after the time limit when $k=16$)}
\label{tab:S}
\end{table*}

\begin{table*}[ht]
\centering
\begin{tabular}[t]{|c|r|r|r|r|r|r|r|r|r|r|r|r|}\hline
\rowcolor{lightgray} $k$ (Solution?)  & 6\YY      & 7\NN    & 8\NN    & 9\NN    &  10\NN & 11\NN  & 12\YY & 13\NN  & 14\NN  &  15\NN \\\hline
BurrTools 0.6.3 & 6 & $<1$ & $<1$ & ? & ? & ? & ? & ? & ? & ? \\\hline
DLX(1st solution) & $<1$   & $<1$ & ?  & ?  & ? & ? & $<1$ & ? & ? & ? \\
DLX(all solutions) & $<1$   & ?  & ?  & ?  & ? & ? & ?  & ? & ? & ? \\\hline
DLZ(1st solution) & $<1$  & $<1$ & $<1$ & $<1$ & $<1$ & 241 & $<1$ & 146 & OF  & OF \\
DLZ(all solutions) & $<1$  & ?  & ?  & ?  & ?  & ? & 6  & ? & ? & ? \\\hline
SCIP 7.0.2    & 1 & 1 & 1 & 5 &  9 & 14 & 69 & 4 & 6 & 19800 \\\hline
clingo 5.4.0  & $<1$ & $<1$ & $<1$ & 1 & 2  & 5 & 4 & 2 & 4 & 5 \\\hline
NaPS 1.02b2   & $<1$ & $<1$ & $<1$ & $<1$    & 1    & 2    &  2   & 7  & 52 & 116 \\\hline\hline
\rowcolor{lightgray} $k$ (Solution?) & 16\NN  & 17\NN   & 18\YY    & 19\NN  & 20\NN    & 21\NN  & 22\NN  & 23\NN  & 24\YY & 25\NN \\\hline
clingo 5.4.0  & 6 & 11 & 9 & 5336 & 41489 & ?  & ? & ? & 1454 & ? \\\hline
NaPS 1.02b2   & 208 & 282  & 113   & 1531 & 116400 & ? & ? & ? & 1675 & ? \\\hline
\end{tabular}
\caption{Time (sec.) for deciding if \drawJ{} is a rep-tile of rep-$k^2$}
\label{tab:J}
\end{table*}

\begin{table*}[ht]
\centering
\begin{tabular}[t]{|c|r|r|r|r|r|r|r|r|r|r|r|r|}\hline
\rowcolor{lightgray} $k$ (Solution?) & 6\NN      & 7\NN    & 8\YY    & 9\YY   &  10\NN & 11\NN  & 12\YY   & 13\YY  & 14\NN  &  15\YY \\\hline
BurrTools 0.6.3 & $<1$ & 960 & 172800 & ? & ? & ? & ? & ? & ? & ? \\\hline
DLX(1st solution) & $<1$   & $<1$ & $<1$ & $<1$ & ? & ? & $<1$ & ? & ? & ? \\
DLX(all solutions) & ?    & ?  & ?  & ?  & ? & ? & ?  & ? & ? & ? \\\hline
DLZ(1st solution) & $<1$  & $<1$ & $<1$ & $<1$ & $<1$ & 102 & $<1$ & 20  & ? & ? \\
DLZ(all solutions) & ?   & ?  & $<1$ & $<1$ & ?  & ? & OF   & OF  & ? & ? \\\hline
SCIP 7.0.2    & 1 & 2 & 43 & 13 & 11 & 259200 & ? & ? & ? & ? \\\hline
clingo 5.4.0  & $<1$ & $<1$ & $<1$ &  1 & 3  & 4  & 11 & 37 & 97 & 372 \\\hline
NaPS 1.02b2   & $<1$ & $<1$ & $<1$ & $<1$ & 2  & 3 & 10 & 14 & 671 & 688 \\\hline\hline
\rowcolor{lightgray} $k$ (Solution?) & 16\YY & 17\YY   & 18(?)    & 19\YY  & 20\YY & 21\YY  & 22(?)  & 23\YY  & 24\YY  & 25\YY\\\hline
clingo 5.4.0  & 244 & 134 & ? & 18022 & 6498  & ?  & ? & ? & ? & ? \\\hline
NaPS 1.02b2   & 316 & 505  & ?  & 7455 & 6249 & 8485 & ? & 47550 & 131900  & 146200 \\\hline
\end{tabular}
\caption{Time (sec.) for deciding if \drawF{} is a rep-tile of rep-$k^2$}
\label{tab:F}
\end{table*}

Comparing to the DLX based on just dancing links, BurrTools implements some more tricks.
The DLZ, which uses not only dancing links but also ZDD, performs more efficiently than DLX, however,
it causes memory overflow when the search space becomes larger.
Comparing to the algorithms based on dancing links, 
the MIP solver SCIP can deal with a larger scale.
We note that we do not need the optimization function of the MIP solver in rep-tile.
When we use SAT-based solvers clingo and NaPS,
the range that can handle is much wider than the other problem solvers.

The details of each experiment are described below, however, there are differences in resources depending on problem solvers.
This is because the authors split up to perform experiments that was good at each tool. 
The difference in computation results due to the difference in resources is considered to be tens to hundreds of times, 
however, considering the scale of the problem that increases exponentially and 
the actual computation results in Tables \ref{tab:S}, \ref{tab:J}, and \ref{tab:F},
it can be seen that the differences of these constant factors do not affect our conclusion.
The following are the details for each experimental environment.

\subsection{Puzzle solvers}

BurrTools 0.6.3\footnote{\url{http://burrtools.sourceforge.net/}} is widely recognized as the standard puzzle solver in the puzzle society.
It supports a variety of grids and also supports 2D and 3D for puzzles that ask to pack a given set of pieces into a given frame (without overlapping or gaps).
According to the web page of BurrTools, it is based on the data structure dancing links proposed by Knuth, who wrote a 270-page textbook \cite{Knuth5}.
Dancing links is a data structure for efficiently performing backtracking in a tree search by depth-first search.
In the literature \cite{Knuth5}, many examples are taken from famous puzzles as applications of backtracking in search trees.
In fact, the polyomino packing puzzle, which is essentially the same as the rep-tile, is also taken up in detail as an example.
In our experiments, the machine used has an Intel Core i5-7300U (2.60GHz) CPU and 8GB of RAM.
It is the limit of analysis for $k=8$, namely, the rep-tile of rep-64 in each pattern.

BurrTools does various tunings internally, however, the details are not public.
For comparison, we first implemented using dancing links as they are.
The machine used has a CPU of Ryzen 7 5800X (3.8GHz) and 64GB of RAM.
The C program for the experiment used DLX1\footnote{See \url{https://www-cs-faculty.stanford.edu/~knuth/programs.html} and \linebreak
\url{https://www-cs-faculty.stanford.edu/~knuth/programs/dlx1.w} for the details.} 
developed by Knuth.
When using DLX1, it turns out that $k=12$ is the limit in terms of finding a solution, and $k=6$ is difficult in terms of finding all solutions.
Next, we tried to speed up the search by combining dancing links with ZDD.
The C program for the experiment used DLX6\footnote{\url{https://www-cs-faculty.stanford.edu/~knuth/programs/dlx6.w}} developed by Knuth.
The word ZDD is an abbreviation for Zero-suppressed Binary Decision Diagram, 
and it is a data structure that shares subtree structures that appear in common in the binary decision tree.
In particular, the memory efficiency is further improved compared to the normal BDD 
by not maintaining the path when the result becomes 0 (see \cite{ZDD} for details).
If ZDD is used in a tree search like our problem, since it is not necessary to repeatedly search the already searched subtree, a significant speedup can be expected.
On the other hand, it is necessary to store all the subtrees once searched in the cache, and hence the memory efficiency is worse than the depth-first search tree.
By speeding up using ZDD, it is possible to achieve up to $k=13$ in the sense of finding a solution, 
and $k=9$ for the F-shape and $k=12$ for the J-shape in the sense of finding all the solutions.
However, when the scale was larger than that, the search could not be completed due to lack of memory.

\subsection{MIP solver}

As the MIP solver, SCIP 7.0.2 \footnote{\url{https://www.scipopt.org/}} was used in this research.
The way of modeling is as introduced in Section \ref{sec:pre}.
SCIP requires a term for optimization, however, it is redundant in our model.
Hence, we minimize the sum of binary variables as a dummy.
Whenever it is feasible, the result comes to $n=k^2$, so it acts as a double check for the feasible solution.

The machine used in the experiment has an Intel Core i5-7300U (2.60GHz) CPU and 8GB of RAM.
Although the results a bit vary, it can be seen that the solvable range is wider than when BurrTools is used.

\subsection{SAT-based solvers}
\label{sec:sat}

Some SAT-based solvers support Pseudo Boolean Constraints (PBs) (see \cite{Knuth5SAT} for details).
All the constraints used in the above MIP solver are within the range of PB except for the optimization term.
Here, the optimization term in the MIP solver was redundant information when finding solutions of the rep tile.
Therefore, when the optimization term is deleted from the constraint descriptions used in the above MIP solver, 
it can be solved by the SAT-based solvers that can handle PBs as they are.

In this research, we used two typical SAT-based solvers for deciding satisfiability;
clingo 5.4.0 \footnote{\url{https://potassco.org/clingo/}} and NaPS 1.02b2 \footnote{\url{https://www.trs.cm.is.nagoya-u.ac.jp/projects/NaPS/}}.
The machine used to run clingo has an Intel Core i7 (3.2GHz) CPU and 64GB of memory, 
and the machine used to run NaPS has a Core i3 (3.8GHz) CPU and 64GB of memory.
Each computation time corresponds to the time for finding the first solution in the other solvers.
Even considering the differences among the execution environments, 
it can be concluded that the range that can be solved by the SAT-based solvers is dramatically expanded compared to the puzzle solver and the MIP solver.

\subsection{How to count the number of solutions}

From the experiments, it was found that the best method for determining the existence of the solution is to use the SAT-based solvers.
The SAT-based solvers used in Section \ref{sec:sat} have the function of finding all solutions in addition to determining whether or not it is satisfiable.
However, it is not practical since it will take time due to the large number of solutions.
On the other hand, the projected model counting solver GPMC\footnote{\url{https://www.trs.cm.is.nagoya-u.ac.jp/projects/PMC/}}
 cannot find a solution for given constraints in CNF, however, the number of solutions can be found at high speed.

Therefore, in order to compute the number of solutions, 
we first determine the satisfiability using NaPS, 
next convert the constraints described in PB to the CNF using the conversion function of NaPS if it is satisfiable.
Then the number of models was counted by GPMC.
(To be more precise, when PBs are converted to CNF, variables other than the binary variable of interest are also generated.
Therefore, the GPMC projection model counting function is used to count only the number of satisfiable assignments to the variable of interest. 
We can count the number of satisfiable solutions by this way.)

Table \ref{tab:rep-tile} summarizes the number of solutions obtained by combining NaPS and GPMC in this way.
The entry written as $>0$ in the table is the entry confirmed that the solution exists using NaPS, 
and the entry that specifically describes the number of solutions is the entry that was successfully counted by GPMC.
The ? mark indicates that any solution could not be found after running NaPS for 2 days.
Here, for the $k=16$ in stair-shape, a solution was found when the time limit was exceeded.

\section{Analysis and New Solutions}
\label{sec:cmp}

As shown in Section \ref{sec:exp}, 
through this research, we were able to compute the number of solutions of rep-tile solutions up to a previously unknown size.
Specifically, in each case, the existence of solutions was determined by NaPS, and the number of solutions was counted by GPMC.
However, although the total number of solutions can be found with this method, the details of the solutions are not clear.
In this section, we observe the solutions by NaPS and the number of solutions by GPMC, referring to the known results, and clarify the details of solutions for some $k$.
As a result, we find new solutions that were not included in the known results at all.
We will look at this in detail for each 6-omino.

\subsection{J-shape 6-omino}

The following property is useful for analysis of J-shape 6-omino (hereafter, we assume $k>1$ to simplify):
\begin{lemma}
\label{lem:Jpair}
Let $k$ be an integer such that J-shape 6-omino is a rep-tile of rep-$k^2$.
Then $k^2$ is an even number and any tiling by $k^2$ copies of J-shape can be dissected into
$k^2/2$ 12-ominoes such that they consist of \drawJRec{} and \drawJX{} (or their mirror images).
\end{lemma}
\begin{proof}
A J-shape piece $P$ is concave. That is, it has a unit square not belonging to $P$ but sharing three edges of $P$.
To cover this unit square by the other J-shape, we have only two ways shown above. This implies the lemma.
\end{proof}
We obtain a corollary by Lemma \ref{lem:Jpair}.
\begin{corollary}
\label{cor:Jpair}
For any odd number $n>1$, a J-shape 6-omino is not a rep-tile of rep-$n$.
Therefore, for any odd number $k>1$, a J-shape 6-omino is not a rep-tile of rep-$k^2$.
\end{corollary}
By Theorem \ref{th:power} and Corollary \ref{cor:Jpair},
it is sufficient to check whether a J-shape 6-omino is a rep-tile of rep-$k^2$ only for even $k$.
Moreover, by Lemma \ref{lem:Jpair}, we can decide if a J-shape 6-omino is a rep-tile 
by checking of tiling using only two 12-omino pieces \drawJRec{} and \drawJX{}.
Using this method, we can complete the computation for $k$ larger than the experiments in Section \ref{sec:exp}.
By combining the arguments with the results in Section \ref{sec:exp},
we obtain the following theorem for the J-shape 6-omino:

\begin{figure}[ht]
\centering
\begin{minipage}[b]{0.38\textwidth}
\centering
\begin{tabular}[c]{c}\includegraphics[width=30mm,bb=0 0 274 284]{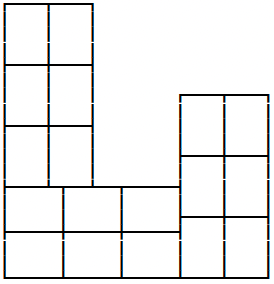}\end{tabular}
\caption{Construction of tiling based on rectangles of size $3\times 4$ for $k=6$} 
\label{fig:k6}
\end{minipage}
%
\begin{minipage}[b]{0.48\textwidth}
\centering
\begin{tabular}[c]{c}\includegraphics[width=60mm,bb=0 0 540 559]{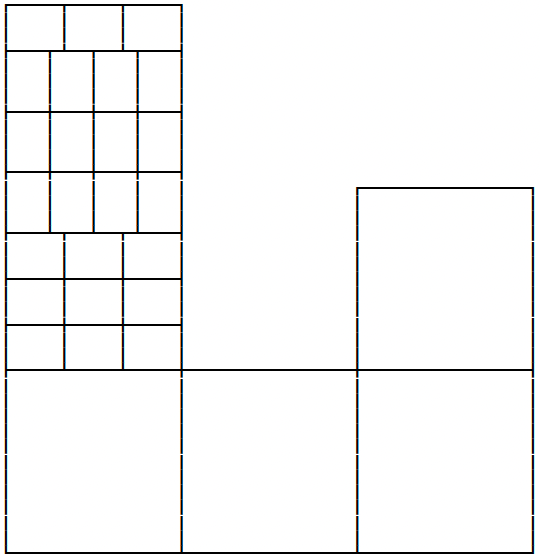}\end{tabular}
\caption{Part of construction of tiling based on rectangles of size $3\times 4$ for $k=12$} 
\label{fig:k12}
\end{minipage}
\end{figure}

\begin{theorem}
\label{th:J}
For a rep-tile of the J-shape 6-omino of rep-$k^2$, we have the following:

\noindent
(0) There exists no rep-tile of rep-$k^2$ for an odd number $k$ (except $k=1$).
There exists no rep-tile of rep-$k^2$ for $k=2,4,8,10,14,16,20,22$.

\noindent
(1) Case $k=6$:
All solutions can be obtained by the following way:
We first dissect the $216$-omino $P$ into $18$ rectangles of size $3\times 4$ as shown in \figurename~\ref{fig:k6} 
and then replace each rectangle by \drawJRec{} or its mirror image.

\noindent
(2) Case $k=12$:
All solutions can be obtained by the following way:
We first dissect the $864$-omino $P$ into $72$ rectangles of size $3\times 4$
and then replace each rectangle by \drawJRec{} or its mirror image.

\noindent
(3) Case $k=18,24$: There are some solutions that contain both \drawJRec{} and \drawJX{}.
\end{theorem}

\begin{proof}
\noindent
(0) We can obtain the results for odd $k$ by Corollary \ref{cor:Jpair}.
The search results by SAT-based solvers in Table \ref{tab:J} give us the results for $k\le 20$.
By Lemma \ref{lem:Jpair}, we perform the search of tiling by copies of two 12-ominoes for larger $k$.
Using NaPS, we confirmed that there is no rep-tile of rep-$k^2$ for $k=22$.

\noindent
(1) Case $k=6$:
The known solutions for the J-shape 6-omino on the web page \cite{RepPage} are based on the arrangement of the rectangle \drawJRec{}.
In fact, when $k=6$, the pattern in which 18 rectangles are arranged (\figurename~\ref{fig:k6}) is shown on the web page.
There are two ways to dissect each rectangle to a pair of two copies of the J-shape 6-omino; \drawJRec{} or its mirror image. 
When $k=6$, the number $262144$ of solutions matches $2^{18}=262144$.
That is, in the case of $k=6$, there are at least $2^{18}$ solutions based on the dissection into the rectangles in \figurename~\ref{fig:k6},
which is equal to the number of solutions actually counted by GPMC.
Since they match, we can guarantee that no other solution exists.

\noindent
(2) When $k=12$, the number of solutions is 545409716939029673955819520.
This number is much larger than $2^{72}$, which is obtained by the same dissection of the case $k=6$.
The reason can be expressed as follows.
We first consider a square corresponding to the unit square of the J-shape polyomino $P$.
In the rep-tile for $k=12$, the square is of size $12\times 12$.
Then we can tile this square by tiling 12 rectangles of size $3\times 4$ in vertical or horizontal.
(We note that we have no such a choice in \figurename~\ref{fig:k6}, and the dissection is uniquely determined.)
Therefore, we have to consider the number of ways of tiling of rectangles in vertical or horizontal.
Moreover, when we consider a large rectangle obtained by joining these squares of size $12\times 12$,
there are variants of tiling of rectangles of size $3\times 4$.
A concrete example is shown in \figurename~\ref{fig:k12}.
In this example, the rectangle of size $12\times 24$ in $P$ is dissected into rectangles of size $12\times 3$, $12\times 12$, and $12\times 9$.
It is not easy to count the number of distinct dissections of $P$ into rectangles of size $3\times 4$.
Therefore, we first count the number of dissections of $P$ for $k=12$ into rectangles of size $3\times 4$ (and $4\times 3$) by GPMC,  which finishes soon.
As a result, the number of ways of dissections is 115495.
Here, we can confirm that $545409716939029673955819520=115495\times2^{72}$.
Therefore, every rep-tile for $k=12$ can be obtained by two steps;
first, dissect $P$ into rectangles of size $3\times 4$ and $4\times 3$,
and then replace each of them by \drawJRec{} or its mirror image.

\begin{figure}[ht]
\centering
\begin{minipage}[b]{0.48\textwidth}
\centering
\begin{tabular}[c]{c}\includegraphics[width=65mm,bb=0 0 645 670]{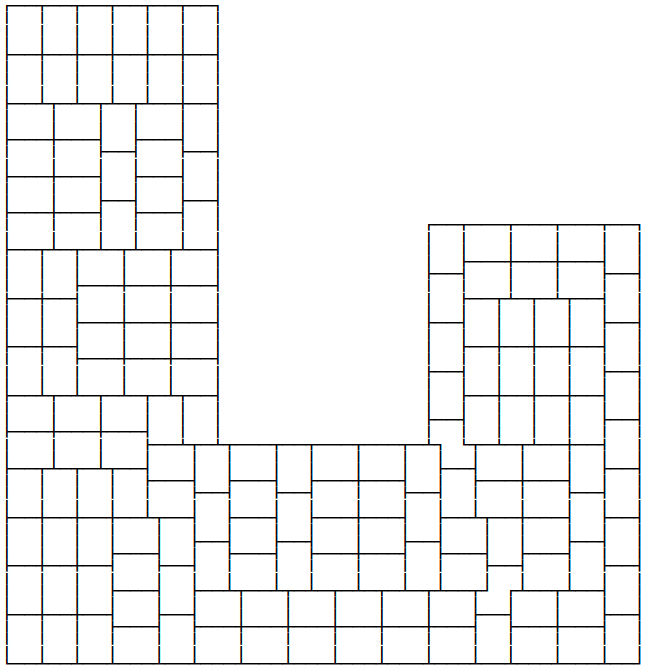}\end{tabular}
\caption{A rep-tile of rep-$18^2$ of J-shape 6-omino that contains \drawJX}
\label{fig:J18x}
\end{minipage}
\begin{minipage}[b]{0.48\textwidth}
\centering
\begin{tabular}[c]{c}\includegraphics[width=70mm,bb=0 0 645 670]{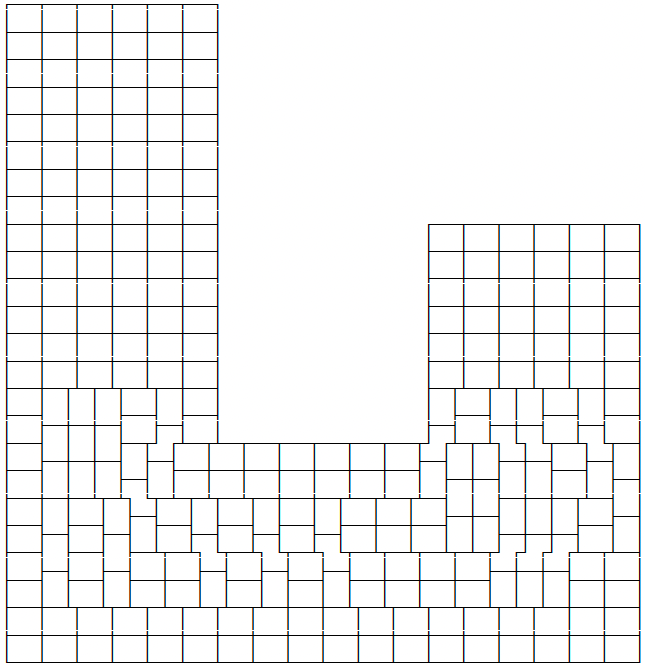}\end{tabular}
\caption{A rep-tile of rep-$24^2$ of J-shape 6-omino that contains \drawJX}
\label{fig:J24x}
\end{minipage}
\end{figure}

\noindent
(3) By Lemma \ref{lem:Jpair}, we can decide if there is a solution that uses \drawJX{} in a tiling by J-shape 6-omino 
by searching using two types of 12-ominoes.
Moreover, when we specify the range of the number of copies of each of two 12-ominoes,
we can decide if there is a solution that contains both \drawJRec{} and \drawJX{}.
As a result, we found that there were such solutions for $k=18,24$; see \figurename~\ref{fig:J18x} and \figurename~\ref{fig:J24x}.
\end{proof}

We note that the solutions that contain \drawJX{} are new solutions not included in previously known results.
So far, in the case $k=18$, there are solutions that contain $x$ copies of \drawJX{} for every even number $x$ from 2 to 46.
There is no such solution when $x\ge 47$. That is, all solutions containing \drawJX{} we found have even number of pairs of this form.
It is not known the details for $k=18$: For example, the number of solutions in the case $k=18$,
whether there exists a solution that contains an odd number copies of \drawJX{}, 
and how many solutions that contain \drawJX{} are not known.
We conjecture that there are solutions that contain \drawJX{} for $k>24$.

\subsection{F-shape 6-omino}

\newcommand{\drawFi}{\begin{tabular}[c]{c}\includegraphics[width=8mm,bb=0 0 68 54]{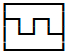}\end{tabular}}

\newcommand{\drawFii}{\begin{tabular}[c]{c}\includegraphics[width=10mm,bb=0 0 81 54]{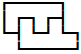}\end{tabular}}

\newcommand{\drawFiii}{\begin{tabular}[c]{c}\includegraphics[width=12mm,bb=0 0 92 54]{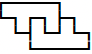}\end{tabular}}

The known rep-tiles of the F-shape 6-omino are a bit complicated, however, the solutions posted on the web page \cite{RepPage} 
are explained as follows: 
We first combine two copies of the F-shape 6-omino to form \drawFi{}, \drawFii{}, and \drawFiii{},
then next arrange them appropriately, and finally place one copy of the F-shape 6-omino if necessary.
In this placement, \drawFi{} is a rectangle, hence replacing it with its mirror image gives us many distinct solutions.

We summarize our results in the following theorem.
Among them, we found new types of solutions that cannot be explained in the way of previously known results for $k=8,15,16,17$.
\begin{theorem}
\label{th:F}
For a rep-tile of the F-shape 6-omino of rep-$k^2$, we have the following:

\noindent
(0) There exists no rep-tile of rep-$k^2$ for $k=2,3,4,5,6,7,10,11,14$.

\noindent
(1) Case $k=8$: 
All solutions can be obtained by the following way:
We first dissect the $384$-omino $P$ in one of the ways shown in \figurename~\ref{fig:Fk8},
and then replace each rectangle by \drawFi{} or its mirror image.

\noindent
(2) Case $k=9$: 
All solutions can be obtained by the following way:
We first dissect the $486$-omino $P$ in one of the ways shown in 
\figurename~\ref{fig:Fk9_1} and \figurename~\ref{fig:Fk9_2}
and then replace each rectangle by \drawFi{} or its mirror image.

\noindent
(3) Case $k=12,13,19,20,21,23,24,25$: 
There exist rep-tiles of rep-$k^2$.
The number of solutions in the case $k=12,13$ can be found in Table \ref{tab:rep-tile}.

\noindent
(4) Case $k=15,16,17$:
There exist rep-tiles of rep-$k^2$ that include the pattern given in \figurename~\ref{fig:F4}.
\end{theorem}

\begin{proof}
\noindent 
(0), (3) We can determine the (non)existence of rep-tiles up to $k=25$ by SAT-based solvers.
By using GPMC, we can count the number of solutions (in the existence case) for each $k$ up to $13$.

\noindent 
(1), (2)
By using NaPS and GPMC, we obtain that the numbers of solutions for $k=8$ and $k=9$ are
1358954496 and 51539607552, respectively.
We then enumerate all non-concave polyominoes that can be obtained by combining 
two or three copies of the F-shape 6-omino, and find all tilings using them.
After that, we count the number of ways of tilings that can be obtained by filling 
each rectangle of size $3\times 4$ by \drawFi{} or its mirror image.
The numbers of tilings should be at most 1358954496 and 51539607552 for $k=8$ and $k=9$, respectively.
In fact, we found that we have already listed all tilings since they are equal in both cases.
The patterns of solutions are listed in \figurename~\ref{fig:Fk8}, \figurename~\ref{fig:Fk9_1}, and \figurename~\ref{fig:Fk9_2}.
We use the all non-concave 18-polyominoes obtained by combining three copies of the F-shape 6-omino, however, 
in fact, only \drawFt{} and \drawFs{} are required to enumerate all solutions for $k=8$ and $k=9$.

\begin{figure*}[ht]
\centering
\includegraphics[width=0.8\textwidth,bb=0 0 1229 431]{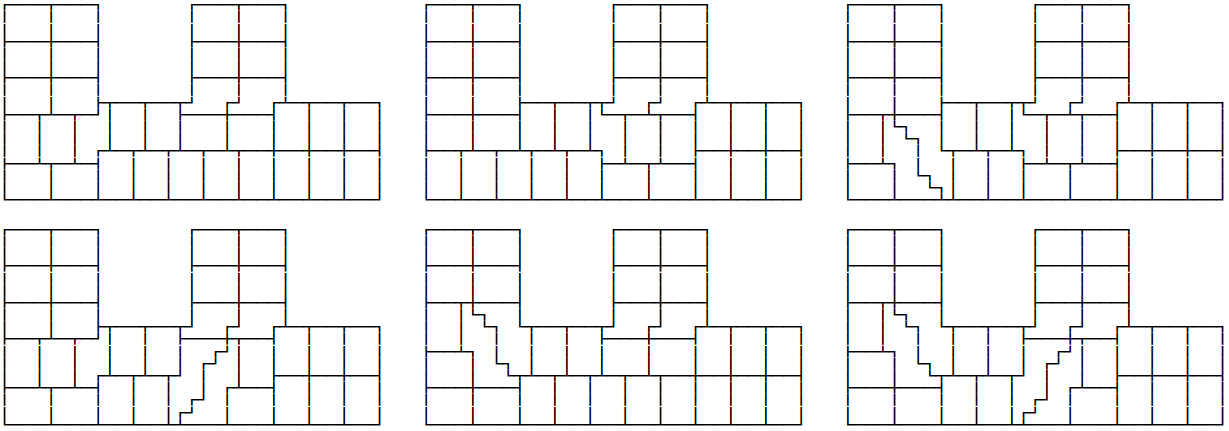}


\caption{All solutions of rep-tiles of rep-$8^2$ for the F-shape 6-omino 
(we can obtain many variants when we fill each rectangle \drawFi{} or its mirror image)}
\label{fig:Fk8}
\end{figure*}

\begin{figure*}[ht]
\centering
\includegraphics[width=0.8\textwidth,bb=0 0 1369 729]{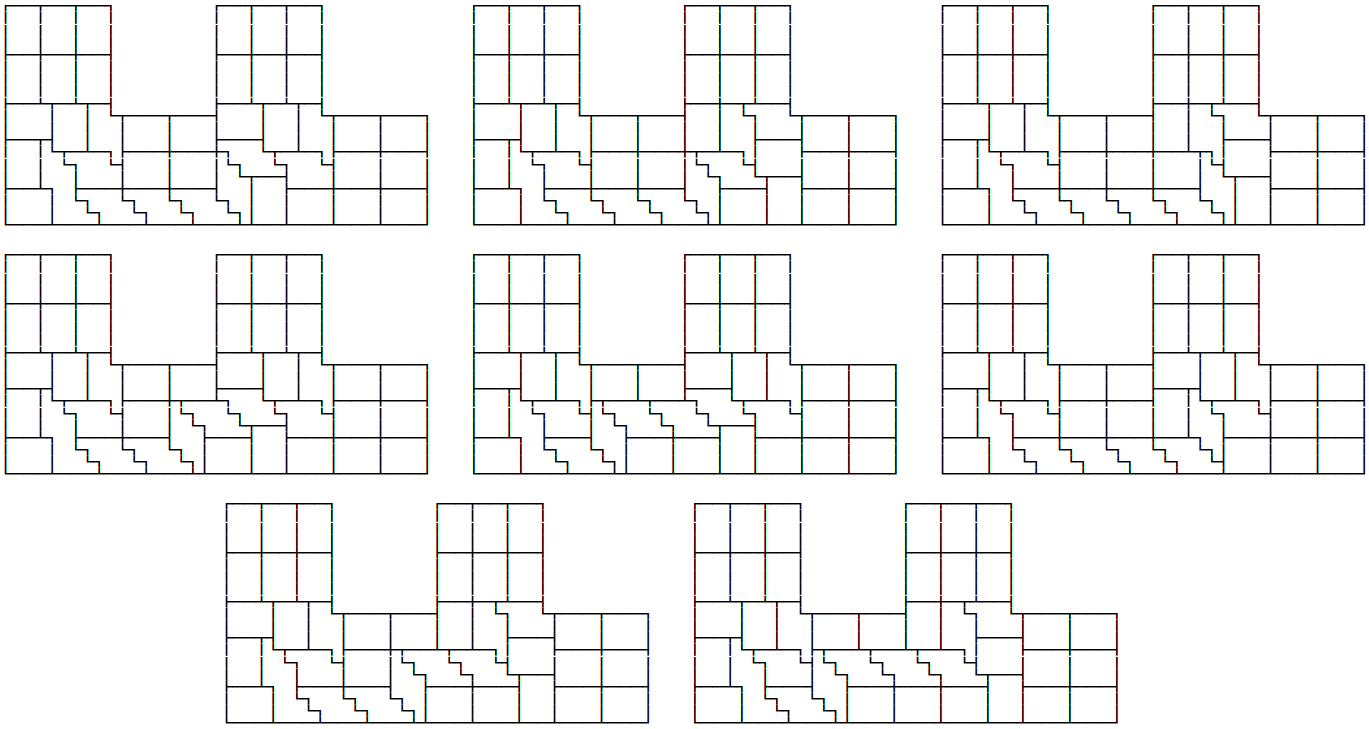}
%
%
%
%
\caption{All solutions of rep-tiles of rep-$9^2$ for the F-shape 6-omino (1/2)}
\label{fig:Fk9_1}
\end{figure*}

\begin{figure*}[ht]
\centering
\includegraphics[width=0.8\textwidth,bb=0 0 1369 479]{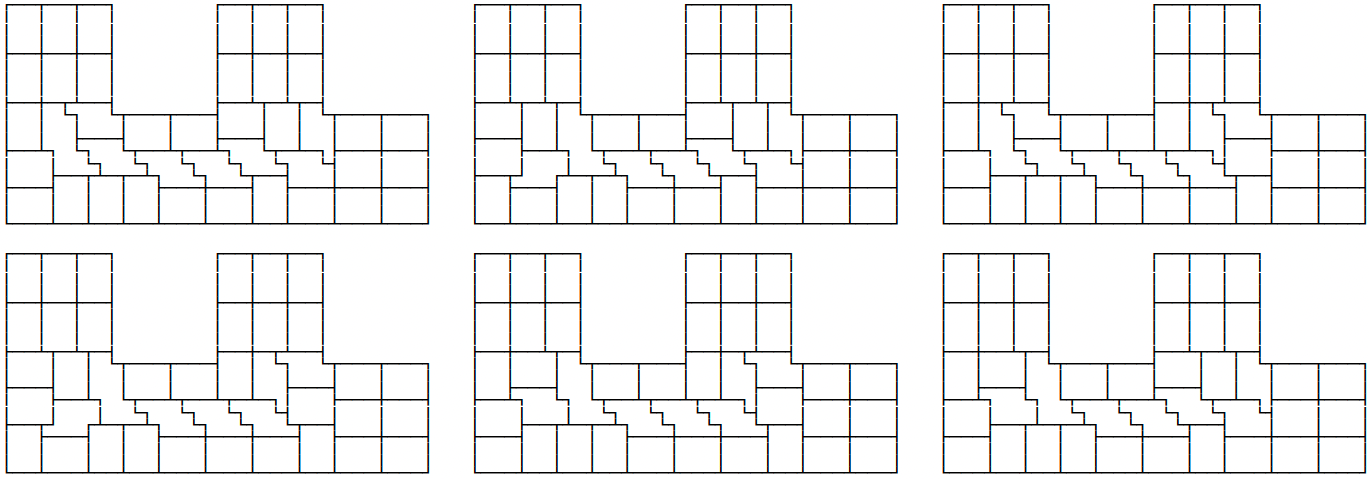}
%
%
%
%
%
%
\caption{All solutions of rep-tiles of rep-$9^2$ for the F-shape 6-omino (2/2)}
\label{fig:Fk9_2}
\end{figure*}

Precisely, when $k=8$, there exist six essentially different dissections.
When we consider replacing each rectangle by \drawFi{} or its mirror image,
we obtain the number of solutions given by \figurename~\ref{fig:Fk8}
is equal to $2^{24} + 2\times 2^{26} + 2^{27} + 2\times 2^{29} = 1358954496$ that coincident with 
the number of solutions obtained by running NaPS and GPMC.
When $k=9$, we have fourteen essentially different dissections.
By considering the numbers of rectangles in these dissections,
the total number of solutions given by \figurename~\ref{fig:Fk9_1} and \figurename~\ref{fig:Fk9_2}
is $8\times 2^{30}+2\times 2^{32}+4\times 2^{33}=51539607552$ that contains all solutions obtained by NaPS and GPMC.

Checking all of these solutions, 
we can confirm that we can construct any rep-tile for $k=9$ 
by combining \drawFi{}, \drawFii{}, and \drawFiii{} 
and add one copy of the F-shape 6-omino if necessary.
Moreover, the last one copy is added to form \drawFt{} or \drawFs{}.
Concretely, \drawFt{} is used in the eight patterns in \figurename~\ref{fig:Fk9_1},
and \drawFs{} is used in the six patterns in \figurename~\ref{fig:Fk9_2}.
That is, when $k=9$, we can construct any solution by tiling some copies of 
\drawFi{}, \drawFii{}, and \drawFiii{} with one copy of \drawFt{} or \drawFs{}.
In other words, these solutions can be represented in the same way of the previously known results.

However, when $k=8$, we cannot construct all solutions in the way of the previously known results.
More precisely, the first two patterns among six patterns in \figurename~\ref{fig:Fk8} can be represented in this way,
however, the next three patterns require to add two copies of the F-shape 6-omino.
Moreover, the last pattern requires to add four copies of the F-shape 6-omino.
That is, among six patterns in \figurename~\ref{fig:Fk8},
there are only two patterns that can be represented in the way of the previously known results
and the other four patterns give us new solutions.
Especially, in the last two patterns in \figurename~\ref{fig:Fk8},
we have to place both copies of \drawFt{} and \drawFs{} after placements of 
copies of \drawFi{}, \drawFii{}, and \drawFiii{}.

\begin{figure}[thb]
\centering
\begin{minipage}[t]{0.2\textwidth}
\centering
\includegraphics[width=0.7\textwidth]{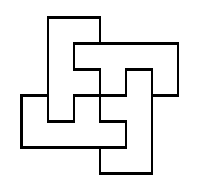}
\caption{A non-concave 24-omino that requires four copies of the F-shape 6-omino (any removal of one or two copies makes concave)}
\label{fig:F4}
\end{minipage}
\hfill
\begin{minipage}[t]{0.7\textwidth}
\centering
\includegraphics[bb=0 0 534 817,width=0.7\textwidth]{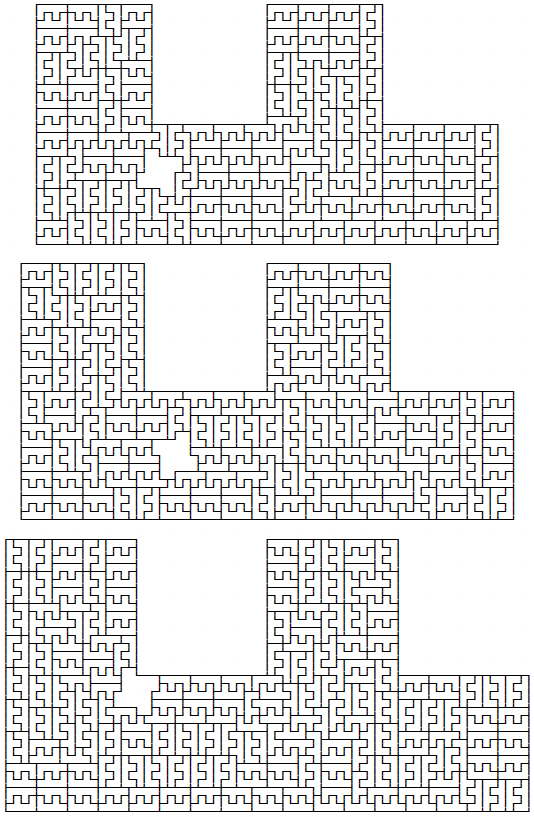}
\caption{Examples of rep-tiles that contain the pattern in \figurename~\ref{fig:F4} for $k=15,16,17$}
\label{fig:windmill}
\end{minipage}
\end{figure}
\end{proof}

\noindent 
(4) In the case of $k=8$ or $k=9$, we can construct all rep-tiles by tiling non-concave polyominoes obtained by 
combining two or three copies of the F-shape 6-omino. 
Then, is this common in all the rep-tiles by the F-shape 6-omino? It is not the case.
We first note that there exist patterns that require four or more copies of the F-shape 6-omino.
A concrete example is given in \figurename~\ref{fig:F4}.
(There are no rep-tile containing such a pattern in the previously known results.)
We searched rep-tiles that require copies of the pattern in \figurename~\ref{fig:F4}
with non-concave polyominoes obtained by combining two or three copies of the F-shape 6-omino. 
Then there are some solutions (\figurename~\ref{fig:windmill})
 containing the pattern in \figurename~\ref{fig:F4} for $k=15,16,17$.
They are completely different rep-tiles from the previously known solutions.

\subsection{Stair-shape 6-omino}

\begin{figure*}[htb]
\centering
\begin{tabular}{ccc}
\includegraphics[width=0.3\textwidth]{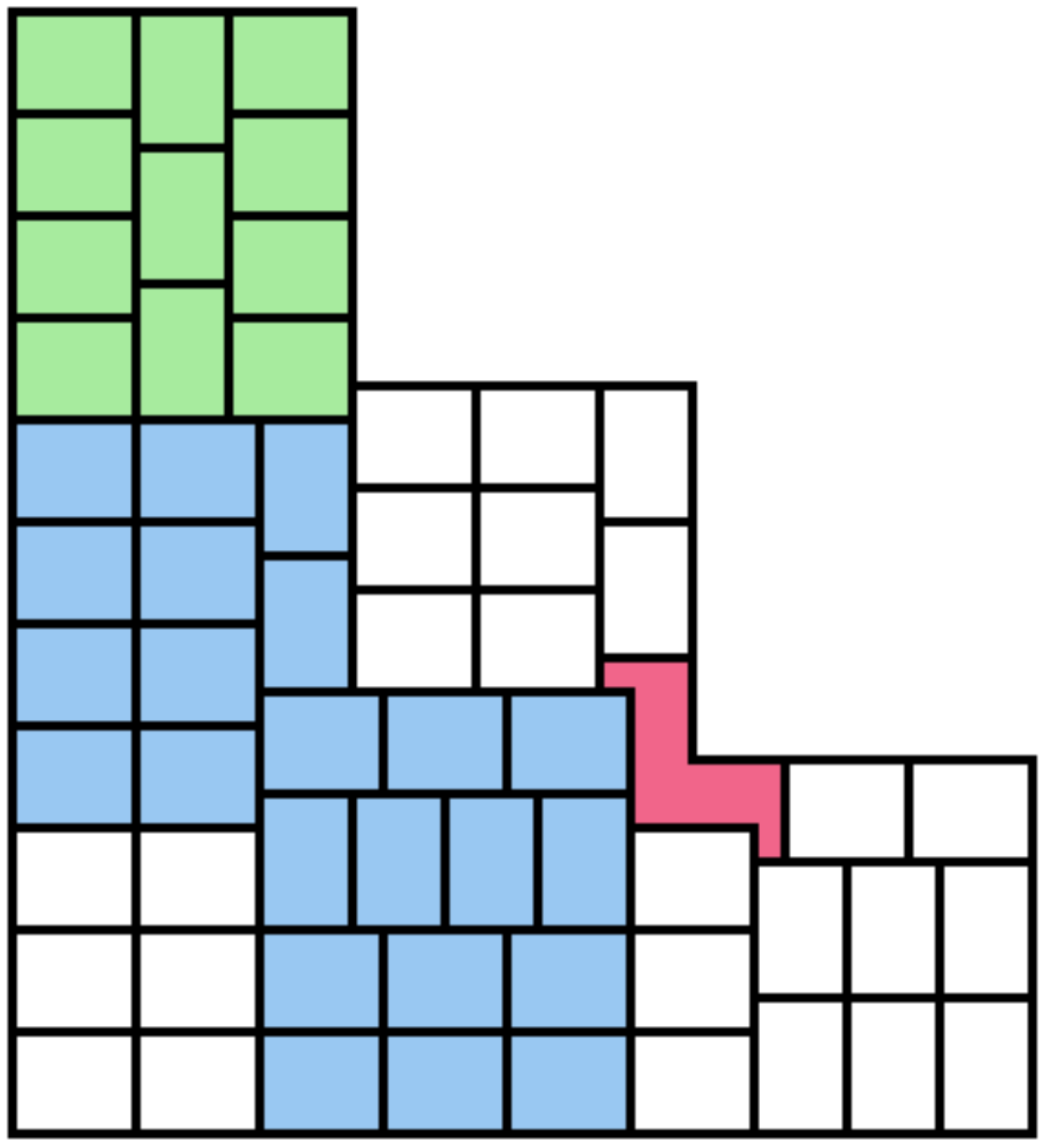} & 
\includegraphics[width=0.3\textwidth]{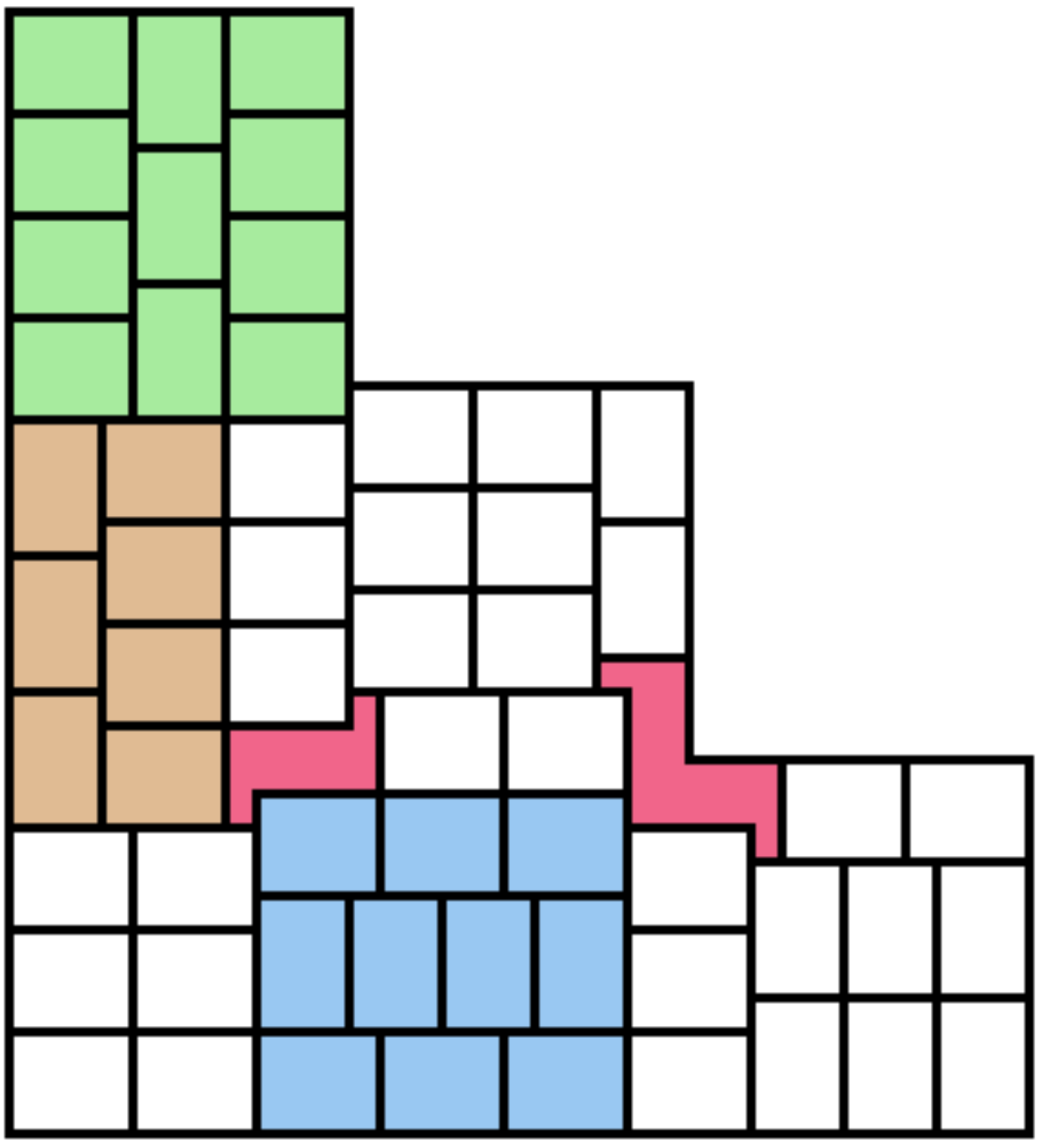} & 
\includegraphics[width=0.3\textwidth]{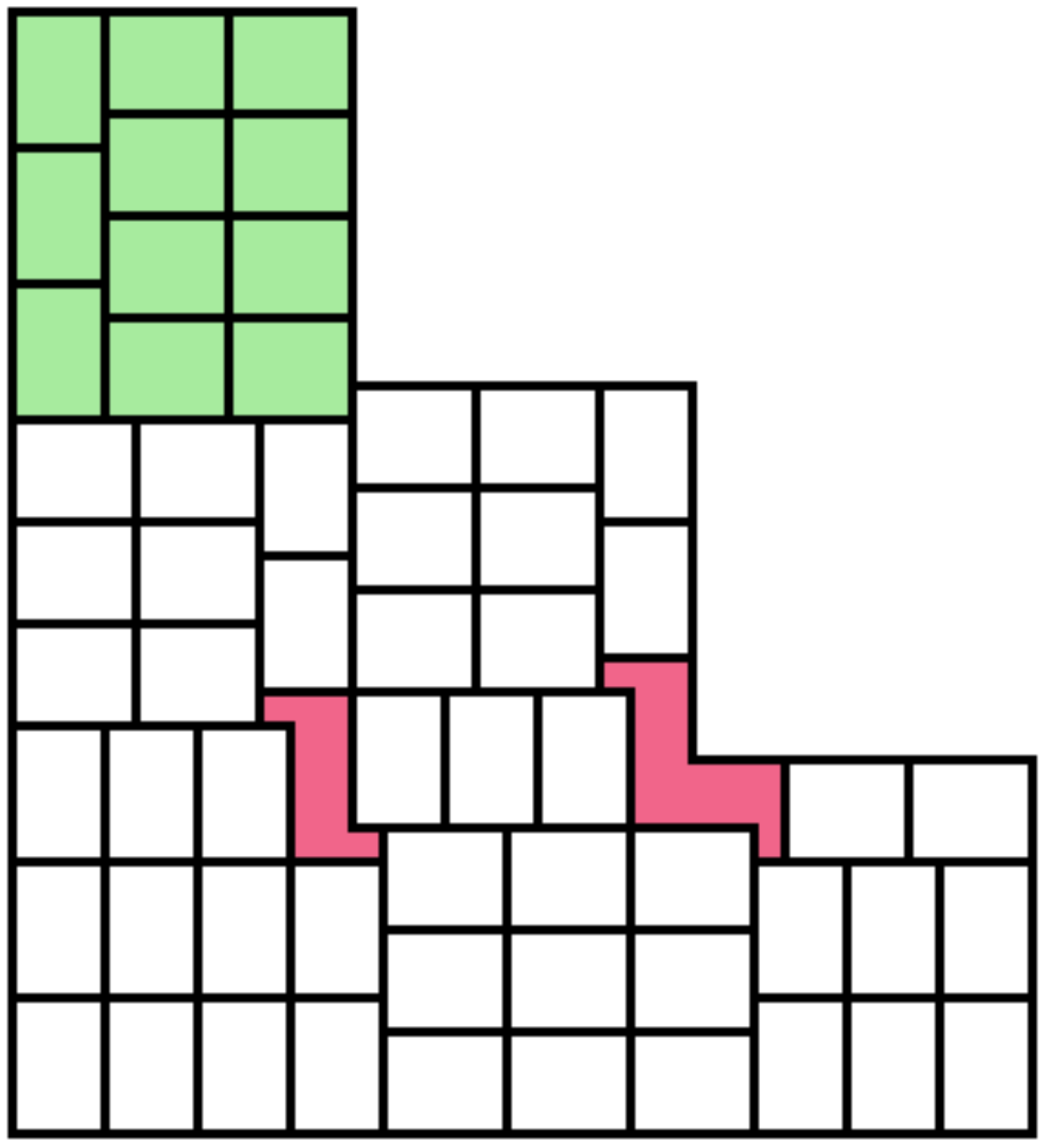} \\
(a) & (b) & (c) 
\end{tabular}
\caption{Three solutions of the stair-shape 6-omino for $k=11$}
\label{fig:revS-size11}
\end{figure*}

Since the stair-shape 6-omino is not concave (in our definition) contrast with the J-shape and F-shape 6-ominoes,
it is difficult to search its rep-tile pattern systematically.
However, by generating the unit patterns obtained by combining a few copies of the stair-shape
to cancel the zig-zag part of it and tiling the copies of these unit patterns,
we succeeded to generate all patterns of solutions for $k=11$.
The results can be summarized as follows:

\begin{theorem}
\label{th:S}
For a rep-tile of the stair-shape 6-omino of rep-$k^2$, we have the following:

\noindent
(0) There exists no rep-tile of rep-$k^2$ for $k=2,3,4,5,6,7,8,9,10,14,15,16,17,18$.

\noindent
(1) Case $k=11$:
All solutions can be obtained by the following way:
We first dissect the $726$-omino $P$ into one of three patterns in \figurename~\ref{fig:revS-size11}.
Then replace each polygon by \drawSz{}, \drawSo{}, or \drawSt{} (or their mirror images).
We note that the previously known results are included in \figurename~\ref{fig:revS-size11}(a),
and the patterns in \figurename~\ref{fig:revS-size11}(b)(c) are new solutions that we found in this research.

\noindent
(2) Case $k=12,13,23,24,25$: 
There exist rep-tiles of rep-$k^2$.
The number of solutions in the case $k=12,13$ can be found in Table \ref{tab:rep-tile}.
\end{theorem}

\begin{proof}
We omit all the cases except $k=11$ since they were obtained by NaPS and GPMC.
(Here we note that $k=16$ is an exception: the solution in this case could not be obtained by the time limit,
however, we could obtain it when we extend the time limit.)
When $k=11$, we perform the search by using three 12-ominoes obtained by \drawSz, \drawSo{}, and \drawSt.
We have three groups by the search.

The first pattern is given in \figurename~\ref{fig:revS-size11}(a):
It uses 59 copies of \drawSz{} and one copy of \drawSt{}.
There are three ways of tiling the left upper green rectangle by 11 rectangles of size $3\times 4$,
and six ways of tiling the blue polygon by 23 rectangles of size $3\times 4$.
(For the latter blue polygon,
there are three ways of tiling of the left upper blue rectangle of size $11\times 12$,
four ways of tiling of the right lower blue rectangle of size $12\times 13$,
and one in common, which implies six ways in total.)
Since we can make a mirror image with respect to the line of 45 degrees,
the total number of solutions in the pattern in \figurename~\ref{fig:revS-size11}(a)
is $2\times 18\times 2^{59}=20752587082923245568$.

The next pattern is given in \figurename~\ref{fig:revS-size11}(b),
which uses 58 copies of \drawSz{}, one copy of \drawSo{}, and one copy of \drawSt{}.
In this case, there are three ways to tile the left upper green rectangle,
two ways to tile the central brown rectangle, and 
three ways to tile the lower blue rectangle.
The last pattern in \figurename~\ref{fig:revS-size11}(c) also 
uses 58 copies of \drawSz{}, one copy of \drawSo{}, and one copy of \drawSt{}.
It has three ways to tile the green rectangle.
In total, the number of solutions in patterns in \figurename~\ref{fig:revS-size11}(b) and 
\figurename~\ref{fig:revS-size11}(c) is $2\times (18+3)\times 2^{58}=12105675798371893248$.

Therefore, when we add all solutions in the patterns in \figurename~\ref{fig:revS-size11}(a)(b)(c),
it makes $42\times 2^{58}+36\times 2^{59}=32858262881295138816$,
which is equal to the number of solutions in Table \ref{tab:rep-tile}.
Therefore, we cover all rep-tiles for $k=11$.
\end{proof}


\section*{Acknowledgments}
This research is partially supported by Kakenhi (17K00017, 18H04091, 20H05964, 21K11757).


\end{document}